\documentclass[conference,dvi]{IEEEtranKV}

\usepackage{graphicx}
\usepackage{citesort}
\usepackage{amssymb}
\usepackage{amsfonts}
\usepackage{amsmath}
\usepackage{epsfig}
\usepackage{color}
\usepackage{fancybox}
\usepackage{textcomp}
\usepackage{multirow}
\usepackage{setspa ce}
\usepackage{psfrag}
\usepackage[ruled,vlined,linesnumbered]{algorithm2e}

\usepackage{mathrsfs}

\newtheorem{Proposition}{Proposition}

    \newcommand{\qh}{{\bf h}}

    \newcommand{\qx}{{\bf x}}
    
    \newcommand{\qz}{{\bf z}}

    \newcommand{\qH}{{\bf H}}
    \newcommand{\qI}{{\bf I}}

    \newcommand{\qQ}{{\bf Q}}

    \newcommand{\qW}{{\bf W}}
    \newcommand{\qX}{{\bf X}}
    \newcommand{\qY}{{\bf Y}}
    \newcommand{\qZ}{{\bf Z}}
    
    \newcommand{\qone}{{\bf 1}}

    \newcommand{\tc}{{\tilde{c}}}

    \newcommand{\tq}{{\tilde{q}}}

    \newcommand{\wtqY}{{\widetilde{\qY}}}
    \newcommand{\wtY}{{\widetilde{Y}}}

    \newcommand{\uqX}{{\underline{\qX}}}

    \newcommand{\bbC}{{\mathbb C}}

    \newcommand{\calN}{{\mathcal N}}
    
    \newcommand{\calT}{{\mathcal T}}

    \newcommand{\calqH}{\boldsymbol{\cal H}}
    \newcommand{\calqX}{\boldsymbol{\cal X}}
    \newcommand{\calqQ}{\boldsymbol{\cal Q}}
    
    \newcommand{\calqZ}{\boldsymbol{\cal Z}}

    \newcommand{\tr}{{\sf tr}}
    \newcommand{\Ex}{{\sf E}}

    \newcommand{\Extr}{\operatornamewithlimits{\sf Extr}}
    
    \newcommand{\mse}{{\sf mse}}

    \newcommand{\sfQ}{{\sf Q}}

    \newcommand{\rmd}{{\rm d}}
    \newcommand{\rmD}{{\rm D}}

    \newcommand{\sfd}{{\sf d}}
    
    \newcommand{\sfj}{{\sf j}}
    \newcommand{\sft}{{\sf t}}

    \newcommand{\sfB}{{\sf B}}
    
    \newcommand{\sfF}{{\sf F}}
    \newcommand{\sfH}{{\sf H}}
    \newcommand{\sfP}{{\sf P}}
    \newcommand{\sfX}{{\sf X}}

    \newcommand{\scP}{\mathscr{P}}

\begin{document}

\title{Joint Channel-and-Data Estimation for Large-MIMO Systems with Low-Precision ADCs}

\author{Chao-Kai~Wen,~Shi~Jin,~Kai-Kit~Wong,~Chang-Jen~Wang,~and~Gang~Wu\thanks{C.-K. Wen and C.-J. Wang are with the Institute of Communications Engineering, National Sun Yat-sen University, Kaohsiung, Taiwan (e-mail: $\rm chaokai.wen@mail.nsysu.edu.tw$). S. Jin is with the National Mobile Communications Research Laboratory, Southeast University, Nanjing 210096, P. R. China. K. Wong is with the Department of Electronic and Electrical Engineering, University College London, UK. G. Wu is with the National Key Laboratory of Science and Technology on Communications, University of Electronic Science and Technology of China, Chengdu 611731, P. R. China.
}
}


\maketitle

\begin{abstract}
The use of low precision (e.g., $1-3$ bits) analog-to-digital convenors (ADCs) in very large multiple-input multiple-output (MIMO) systems is a  technique to reduce cost and power consumption. In this context, nevertheless, it has been shown that the training duration is required to be {\em very large} just to obtain an acceptable channel state information (CSI) at the receiver. A possible solution to the quantized MIMO systems is joint channel-and-data (JCD) estimation. This paper first develops an analytical framework for studying the quantized MIMO system using JCD estimation. In particular, we use the Bayes-optimal inference for the JCD estimation and realize this estimator utilizing a recent technique based on approximate message passing. Large-system analysis based on the replica method is then adopted to derive the asymptotic performances of the JCD estimator. Results from simulations confirm our theoretical findings and reveal that the JCD estimator can provide a significant gain over conventional pilot-only schemes in the quantized MIMO system.
\end{abstract}


\section*{\sc I. Introduction}
Very large multiple-input multiple-output (MIMO) or ``massive MIMO'' systems \cite{Marzetta-10TW} are widely considered as a key technology for 5G wireless communications networks \cite{Larsson-14COMMag,Andrews-14JSAC}. Such systems promote the use of a very large number of antennas at the base station (BS) (e.g., hundreds or thousands) to serve a number of user terminals (e.g., tens or hundreds) in the same time-frequency resource. Nonetheless, the high dimensionality greatly increases hardware cost and power consumption. This motivates the study of MIMO systems with very low precision (e.g., $1-3$ bits) analog-to-digital convenors (ADCs).

\begin{figure}
\begin{center}
\resizebox{3.5in}{!}{%
\includegraphics*{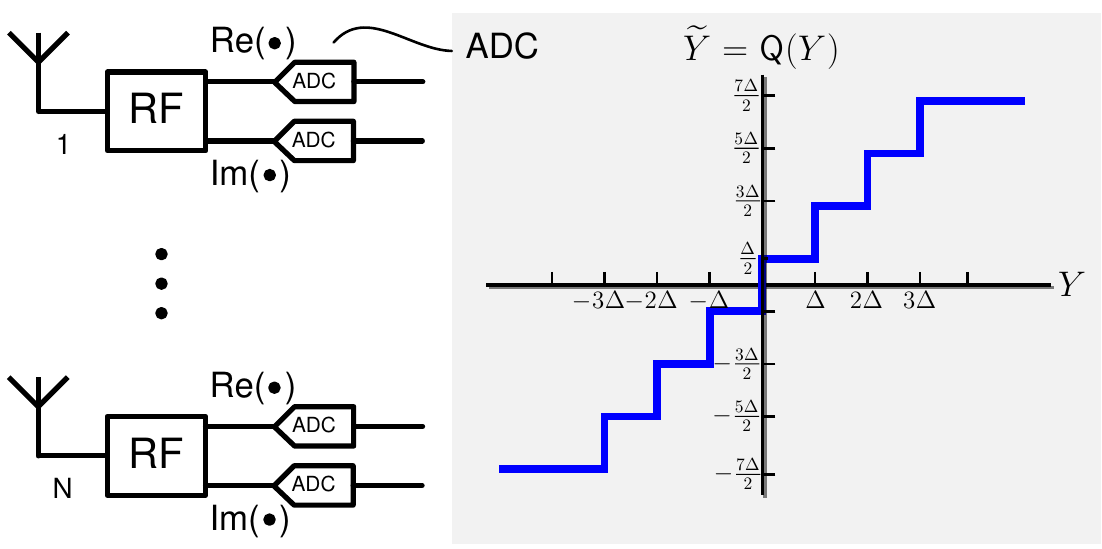} }%
\caption{The quantized MIMO system.}\label{fig:QuantizedMIMO}
\end{center}
\end{figure}

Several aspects of low precision ADCs have been studied in the literature for single-input single-output (SISO) channels \cite{Singh-09TCOM} and more recently MIMO channels \cite{Mo-14ArXiv} and references therein. In this paper, our focus is on signal detection for the quantized MIMO systems where each receiving antenna is equipped with a very low precision ADC. Prior work in this direction covered code-division multiple-access (CDMA) systems \cite{Nakamura-08ISITA}, massive MIMO systems \cite{Nakamura-08ISITA,Mezghani-10ISIT,Risi-14ArXiv,Mo-14ArXiv,Wang-15TWCOM}, distributed antenna systems (DASs) \cite{Choi-14ArXiv}, and compressed sensing \cite{Xu-14JSM}. However, most previous work assumed perfect channel state information at the receiver (CSIR). In particular, \cite{Risi-14ArXiv} revealed that in a MIMO system with one-bit ADC, to achieve the same performance as the full CSI case we have to use a very long training sequence (above $50$ times the number of users). Clearly, the assumption of perfect CSIR becomes quite controversial particularly for quantized MIMO systems. The requirement of long training sequence motivates us to consider joint channel-and-data (JCD) estimation in which estimated
payload data are utilized to aid channel estimation.
A major advantage of JCD estimation is that relatively few pilot symbols
are required to achieve the equalization channel and data estimation performances.

Although performance enhancement by using this technique is expected, we are not aware of any study for the \emph{quantized} MIMO systems using JCD estimation.\footnote{In the context of \emph{unquantized} MIMO system, several aspects of the JCD estimation have already been widely studied, see e.g., \cite{Takeuchi-13TIT,Ma-14TSP,Wen-15ICC}.} In this paper, we take the important first step to analyze the achievable performance of the quantized MIMO system using JCD estimation. To this end, we use the Bayes-optimal inference for JCD estimation as this approach provides the minimal mean-square-error (MSE) with respect to (w.r.t.) the channels and payload data. However, the complexity for carrying out the Bayes-optimal JCD estimator appears prohibitive. To address this issue, we use a variant version of belief propagation (BP) to approximate the marginal distributions of each data and channel components. In particular, we modify the bilinear generalized approximate message passing (BiG-AMP) scheme in \cite{Parker-14TSP} and adapt it to the quantized MIMO system. We refer to the proposed method as GAMP-based JCD. Furthermore, by applying large-system analysis based on the replica method from statistical physics, we provide the theoretical performances for the Bayes-optimal JCD estimator.\footnote{In this paper, the Bayes-optimal JCD estimator is regarded as the \emph{theoretical} optimal estimator, while the GAMP-based JCD algorithm can be thought of as a \emph{practical} algorithm to approximate the theoretical optimal estimator.}  Simulations are used to verify the efficiency of the proposed algorithm and the accuracy of our analysis.

\section*{\sc II. System Model}
We consider a block-fading uplink channel with $K$ transmit antennas and $N$ receive antennas, in which the channel remains constant over $T$ consecutive symbol-intervals (i.e., a block). The received signal $\qY \in \bbC^{N \times T}$ over the block interval can be written in matrix form as
\begin{equation} \label{eq:sys}
    \qY = \frac{1}{\sqrt{K}}\qH \qX + \qW = \qZ + \qW,
\end{equation}
where $\qX \in \bbC^{K \times T}$ denotes the transmit symbols in the block,  $\qH \in \bbC^{N \times K}$ is the matrix containing the fading coefficients associated to the channels between the transmit antennas and the receive antennas, $\qW \in \bbC^{N \times T}$ represents the additive temporally and spatially white Gaussian noise with zero mean and element-wise variance $\sigma_{w}^2$, and we define $\qZ \triangleq \frac{1}{\sqrt{K}}\qH \qX$.

In the quantized MIMO system, as shown in Figure \ref{fig:QuantizedMIMO}, each received signal component $Y^{ij}$, $1 \leq i \leq N$, $1 \leq j \leq T$ are quantized separately into a finite set of prescribed values by a $\sfB$-bit quantizer $\sfQ_{c}$. The resulting quantized signals can read
\begin{equation} \label{eq:qsys}
    \wtqY = {\sfQ_{c}\left(\qZ + \qW\right)}.
\end{equation}
Specifically, each complex-valued quantizer $\sfQ_{c}(\cdot)$ is defined as $\widetilde{Y}^{ij} = \sfQ_{c}(Y^{ij}) \triangleq \sfQ({\rm Re}\{ Y^{ij}\}  ) + \sfj
\sfQ({\rm Im}\{ Y^{ij} \} )$, i.e., the real and imaginary parts are quantized separately. The real-valued quantizer $\sfQ$ maps a real-valued input to one of the $2^\sfB$ bins, which are characterized by the set of ${2^\sfB-1}$ thresholds $[r_1,r_2,\dots,r_{2^{\sfB}-1}]$, such that $ -\infty < r_1 < r_2 < \cdots < r_{2^{\sfB}-1} < \infty$. For notational consistence, we define $r_0 = -\infty$ and $r_{2^{\sfB}} = \infty$. The output $\widetilde{Y}$ is assigned a value in $(r_{b-1}, \,r_{b}]$ when the quantizer input $Y$ falls in the interval $(r_{b-1}, \,r_{b}]$ (namely, the $b$-th bin). For example, the threshold of a typical uniform quantizer with the quantization step-size $\Delta$ is given by
\begin{equation}
    r_b = {\left( -2^{\sfB-1} + b \right)} \Delta, \mbox{~~for~~} b=1,\dots,2^{\sfB}-1,
\end{equation}
and the quantization output is assigned the value ${r_b - \frac{\Delta}{2}}$ when the input falls in the $b$-th bin. Figure \ref{fig:QuantizedMIMO} shows an example of the $3$-bit uniform quantizer.

Since the channel matrix $\qH$ needs to be estimated at the receiver, we make the first $T_{1}$ symbols of the block of $T$ symbols serve as pilot sequences. The remaining $T_2= T-T_1$ symbols are used for data transmissions. This setting is equivalent to partitioning $\qX$ as $\qX = [\qX_{1} \, \qX_{2}]$ with $\qX_{1}  \in \bbC^{K\times T_1}$ and $\qX_{2}  \in \bbC^{K \times T_2}$. The training and data phases are referred to as $\sft$-phase and $\sfd$-phase, respectively. We assume that the matrix $\qX_{1}$ (or $\qX_{2}$) is composed of independent and identically distributed (i.i.d.) random variables $\sfX_1$ ($\sfX_{2}$) generated from a known probability distribution $\sfP_{\sfX_{1}}$ (or $\sfP_{\sfX_{2}}$), i.e.,
\begin{equation} \label{eq:proX}
 \sfP_{\sfX}(\qX) = \sfP_{\sfX_{1}}(\qX_{1}) \sfP_{\sfX_{2}}(\qX_{2})
\end{equation}
with ${\sfP_{\sfX_{1}}(\qX_{1}) = \prod_{i,j} \sfP_{\sfX_{1}}(X_{1}^{ij})}$, $\sfP_{\sfX_{2}}(\qX_{2}) = \prod_{i,j} \sfP_{\sfX_{2}}(X_{2}^{ij})$.

Since the pilot and data symbols should appear on constellation points uniformly, the ensemble averages of $\{X_1^{ij}\}$ and $\{X_2^{ij}\}$ are assumed to be zero. In addition, we let $\sigma_{x_1}^2$ and $\sigma_{x_2}^2$ be the transmit powers during the $\sft$-phase and $\sfd$-phase, respectively, i.e., $\Ex\{ |\sfX_1^{ij}|^2\} = \sigma_{x_1}^2$ and $\Ex\{|\sfX_2^{ij}|^2\} = \sigma_{x_2}^2$.

Similarly, we assume that each entry $H^{ij}$ is drawn from the complex Gaussian distribution $\calN_{\bbC}(0,\sigma_{h}^2)$, where $\sigma_{h}^2$ indicates the large scale fading factor.\footnote{For ease of notation, we consider the case where all the transmits have the same large-scale fading factor but our results can be easily extended.} Let $\sfP_{\sfH}(H^{ij}) \equiv \calN_{\bbC}(0,\sigma_{h}^2)$. Then we have
\begin{equation} \label{eq:proH}
 \sfP_{\sfH}(\qH) = \prod_{ij} \sfP_{\sfH}(H^{ij}).
\end{equation}

\section*{\sc III. JCD Estimation}
Our focus is on the setting where the receiver knows the distributions of $\qH$ and $\qX$ but not their realizations. In the conventional pilot-only scheme, the receiver first uses $\qY_1$ and $\qX_1$ to generate an estimate of $\qH$ and then uses the estimated channel for estimating the data $\qX_2$ from $\qY_2$ \cite{Risi-14ArXiv}. In contrast to the pilot-only scheme, we consider JCD estimation, where the BS estimates both $\qH$ and $\qX_2$ from $\widetilde{\qY}$ given $\qX_1$. We treat the problem in the framework of Bayesian inference \cite{Poor-94BOOK}.

To this end, we first define the likelihood, which is the distribution of the received signals under (\ref{eq:qsys}) conditional on the unknown parameters, as:
\begin{equation} \label{eq:likelihood}
    \sfP(\wtqY|\qH,\qX) = \prod_{i=1}^{N} \prod_{j=1}^{T} {\sfP_{\sf out}\left( \wtY^{ij} \Big| Z^{ij}\right)},
\end{equation}
where
\begin{multline} \label{eq:likelihood_each}
    \sfP_{\sf out}\left( \wtY \Big| Z \right)
    = \left( \frac{1}{\sqrt{ \pi \sigma_{w}^2}} \int_{r_{b-1}}^{r_b} \rmd y e^{-\frac{(y - {\rm Re}(Z) )^2}{\sigma_{w}^2}} \right)  \\
      \times \left( \frac{1}{\sqrt{ \pi \sigma_{w}^2}} \int_{r_{b'-1}}^{r_{b'}} \rmd y e^{-\frac{(y - {\rm Im}(Z) )^2}{\sigma_{w}^2}} \right)
\end{multline}
as ${\rm Re}(\wtY)$ and ${\rm Im}(\wtY)$ fall in the $b$-th and the $b'$-th bins, respectively, i.e, ${\rm Re}(\wtY) = r_b$ and ${\rm Im}(\wtY) = r_{b'}$. Let $\Phi(x) \triangleq \int_{-\infty}^{x} \rmD z$ with $\rmD z \triangleq \frac{1}{\sqrt{2\pi}}e^{-\frac{z^2}{2}} \rmd z$. Then (\ref{eq:likelihood_each}) has the following closed-form expression
\begin{equation}
    \sfP_{\sf out}\left( \wtY \Big| Z \right)
    = \Psi_b\big({\rm Re}(Z)\big) \Psi_{b'}\big({\rm Im}(Z)\big),
\end{equation}
where
\begin{equation}
    \Psi_b(x) \triangleq \Phi\left( \frac{\sqrt{2}(r_b- x)}{\sigma_{w}}\right) - \Phi\left( \frac{\sqrt{2}(r_{b-1}- x)}{\sigma_{w}}\right).
\end{equation}
The prior distributions of $\qH$ and $\qX$ are given by (\ref{eq:proH}) and (\ref{eq:proX}), respectively. Then the posterior probability can be computed according to Bayes' rule:
\begin{equation} \label{eq:posteriorPr}
    \sfP(\qH,\qX|\qY) = \frac{\sfP(\wtqY|\qH,\qX)\sfP_{\sfH}(\qH)\sfP_{\sfX}(\qX)}{\sfP(\wtqY)},
\end{equation}
where $\sfP(\wtqY) = \int_{\qH} \int_{\qX} \rmd\qH\rmd\qX \sfP(\wtqY|\qH,\qX)\sfP_{\sfH}(\qH)\sfP_{\sfX}(\qX)$ is the marginal likelihood.

Given the posterior probability, an estimate for $X^{ij}$ can be obtained by the posterior mean
\begin{equation} \label{eq:estX}
    \widehat{X}^{ij} = \int \rmd X^{ij} \scP(X^{ij}) X^{ij},
\end{equation}
where
\begin{equation}
    \scP(X^{ij}) = \int_{\qH} \int_{\qX\setminus X^{ij}} \rmd\qH \rmd\qX \, \sfP(\qH,\qX|\wtqY)
\end{equation}
is the marginal posterior probability of $X^{ij}$. Here, the notation $\int_{\qX\setminus X^{ij}} \rmd\qX $ denotes the integration over all the variables in $\qX$ except for $X^{ij}$. If the MSE of an estimate $\hat{\qX}$ w.r.t.~$\qX$ is defined as
\begin{equation} \label{eq:mseX}
    \mse(\qX_t|\wtqY) = \int_{\qH}\int_{\qX}\rmd\qH\rmd\qX \sfP(\qH,\qX|\wtqY) \| \hat{\qX}_t - \qX_t\|_{\sfF}^2,
\end{equation}
for ${t=1,2}$, then the posterior mean estimator (\ref{eq:estX}) gives the minimum MSE (MMSE) \cite{Poor-94BOOK}. Notice that given a \emph{known} pilot matrix $\uqX_1$, i.e., ${\sfP_{\sfX_{1}}(\qX_{1}) = \delta(\qX_{1}-\uqX_{1})}$, we can easily obtain $\widehat{X}_1^{ij} = \underline{X}_1^{ij}$ from (\ref{eq:estX}). In this case, we have $\mse(\qX_1|\wtqY) = 0$.

Similarly, the Bayes estimate of $H^{ij}$ is given by
\begin{equation} \label{eq:estH}
    \widehat{H}^{ij} = \int \rmd H^{ij} \scP(H^{ij}) H^{ij}
\end{equation}
where ${\scP(H^{ij}) = \int_{\qH\setminus H^{ij}} \int_{\qX} \rmd\qH \rmd\qX  \sfP(\qH,\qX|\wtqY)}$ denotes the marginal posterior probability of $H^{ij}$. The estimate $\widehat{H}$ also minimizes the MSE
\begin{equation} \label{eq:mseH}
    \mse(\qH|\wtqY) = \int_{\qH}\int_{\qX}\rmd\qH\rmd\qX \, \sfP(\qH,\qX|\wtqY) \| \hat{\qH} - \qH\|_{\sfF}^2.
\end{equation}
Hereafter, we will refer to (\ref{eq:estX}) and (\ref{eq:estH}) as the Bayes-optimal estimator.

Although the Bayes-optimal estimator provides the MMSE estimates, direct computations of (\ref{eq:estX}) and (\ref{eq:estH}) are intractable due to high-denominational integrals involved in the marginal posteriors $\scP(X^{ij})$ and $\scP(H^{ij})$. In \cite{Pearl-1988BOOK}, BP provides a practical alternative to approximate these marginal posteriors. In the recent compressed sensing literature, the Bayesian framework in combination with a BP algorithm has given rise to the so-called approximate message passing (AMP) algorithm \cite{Donoho-09PNAS} and the generalized AMP (GAMP)  \cite{Rangan-10ArXiv}. Applying this development to our context of the MIMO system means that when $\qH$ is perfectly known, GAMP can provide a tractable way to approximate the marginal posteriors $\scP(X^{ij})$'s. Remarkably, it has proved that the approximations become exact in the large-system limit for dense matrix $\qH$ with sub-Gaussian entries. More recently, Parker {\em et al.}~in \cite{Parker-14TSP} applied the same strategy of GAMP to the problem of reconstructing matrices from bilinear noisy observations (i.e., reconstructing $\qH$ and $\qX$ from $\qY$), which is referred to as bilinear GAMP (BiG-AMP). The BiG-AMP scheme can be applied to tackle the Bayes-optimal JCD estimator and we can adapt it to be used in the quantized MIMO setting. We call the developed algorithm GAMP-based JCD algorithm. Due to space limitations, we remove details of the algorithm development in this paper while we will show its simulation performances later in Section V.

\section*{\sc IV. Performance Analysis}
knowing the theoretical lower bound of the estimate is useful to assess any developed algorithm. Therefore, in this section, our objective is to derive analytical results for the average MSEs of $\qX_2$ and $\qH$ for the Bayes-optimal JCD estimator, i.e., $\mse_{X_{t}} \triangleq\Ex\{ \mse(\qX_t|\wtqY) \}$ and $\mse_{H} \triangleq \Ex\{ \mse(\qH|\wtqY) \}$. Our analysis investigates the high-dimensional regime where $N, K, T \to \infty$ but the ratios $N/K= \alpha$, $T/K = \beta$, $T_t/K = \beta_t$, for $t=1,2$ are fixed and finite. For convenience, we simply use $K \rightarrow \infty$ (or refer to as the large-system limit) to denote this high-dimension limit. Following the argument of \cite{Krzakala-13ISIT,Kabashima-14ArXiv}, it can be shown that $\mse_{X_{t}}$ and $\mse_{H}$ are saddle points of the average free entropy
\begin{equation}\label{eq:FreeEn}
    \Phi \triangleq \frac{1}{K^2}\Ex_{\wtqY}\left\{\log \sfP(\wtqY) \right\},
\end{equation}
where $\sfP(\wtqY)$ denotes the marginal likelihood in (\ref{eq:posteriorPr}), namely the partition function. The major difficulty in computing (\ref{eq:FreeEn}) is the expectation of the logarithm of $\sfP(\wtqY)$, which, nevertheless, can be facilitated by rewriting $\Phi$ as \cite{Nishimori-01BOOK}
\begin{equation}\label{eq:LimF}
\Phi = \frac{1}{K^2} \lim_{\tau\rightarrow 0}\frac{\partial}{\partial \tau}\log \Ex_{\wtqY}\left\{\sfP^{\tau}(\wtqY)\right\}.
\end{equation}
The expectation operator is moved inside the log-function. We first evaluate $\Ex_{\wtqY}\{\sfP^{\tau}(\wtqY)\}$ for an integer-valued $\tau$, and then generalize the result to any positive real number $\tau$. This technique is called the replica method, and has been widely adopted in the field of statistical physics \cite{Nishimori-01BOOK} and information theory literature, e.g.,  \cite{Tanaka-02IT,Moustakas-03TIT,Guo-05IT,Muller-03TSP,Wen-07IT,Hatabu-09PRE,Takeuchi-13TIT,Girnyk-14TWC}. Under the assumption of replica
symmetry (RS), the following results are obtained.

\begin{Proposition} \label{Pro_MSE}
As $K \to \infty$, the asymptotic MSEs w.r.t.~$\qX_{t}$ and $\qH$ are associated with the MSEs for the scalar Gaussian channels:
\begin{align}
 Y_{\tq_{H}} &= \sqrt{\tq_{H}} H + W_{H}, \label{eq:scalCh_H} \\
 Y_{\tq_{{X_{t}}}} &=  \sqrt{\tq_{{X_{t}}}} X_{t} + W_{X}, \label{eq:scalCh_X}
\end{align}
where $W_{H},W_{X} \sim \calN_{\bbC}(0,1)$ are independent of $H \sim P_{H}$ and $X_{t} \sim  P_{X_{t}}$. Here, the parameters $\tq_{H}$ and $\tq_{{X_{t}}}$ are the solutions to the set of fixed-point equations
\begin{subequations} \label{eq:fxiedPoints}
\begin{align}
 \tq_{H} & = \sum_{t=1}^{2} \beta_t q_{X_{t}} \chi_t, \label{eq:asyVarH} \\
 \tq_{X_{t}} & = \alpha q_{H}  \chi_t,
 \label{eq:asyVarX} \\
 q_{H} & = c_{H} - \mse_{H}, \label{eq:qH} \\
 q_{X_t} & = c_{X_{t}} - \mse_{X_{t}}. \label{eq:qX}
\end{align}
\end{subequations}
where $\mse_{H} = \Ex\{|H - \Ex\{ H | Y_{\tq_{H}} \} |^2\}$ and $\mse_{X_{t}} = \Ex\{|X_{t}- \Ex\{ X_{t} | Y_{\tq_{{X_{t}}}} \}|^2\}$ are the asymptotic MSEs w.r.t.~$\qX_{t}$ and $\qH$, respectively, $c_{X_{t}} \triangleq \Ex\{|X_{t}|^2\} = \sigma_{x_t}^2$, $c_{H} \triangleq \Ex\{|H|^2\} = \sigma_{h}^2$, and
\begin{equation}
 \chi_t \triangleq \sum_{b=1}^{2^{\sfB}-1} \int \rmD v \frac{\Big(\Psi'_{b}\left(\sqrt{q_{H}q_{X_{t}}} v \right)\Big)^2}{\Psi_{b}\left(
 \sqrt{q_{H}q_{X_{t}}} v \right)},
\end{equation}
with
\begin{multline}\label{eq:Pout}
 \Psi_{b}(V_t)
\triangleq \Phi\left(\frac{ \sqrt{2}r_{b}-V_t}{\sqrt{\sigma_{w}^2 + c_{H}c_{X_t}-q_{H}q_{X_t}}} \, \right)\\
\quad - \Phi\left(\frac{ \sqrt{2}r_{b-1}-V_t}{\sqrt{\sigma_{w}^2 + c_{H}c_{X_t}-q_{H}q_{X_t}}} \, \right),
\end{multline}
and
\begin{multline}
\Psi'_{b}(V_t)\triangleq \frac{\partial \Psi_{b}(V_t)}{\partial V_t}\\
=  \frac{ e^{-\frac{(\sqrt{2}r_{b}-V_t)^2}{2(\sigma_{w}^2 + c_{H}c_{X_t}-q_{H}q_{X_t})}} -e^{-\frac{(\sqrt{2}r_{b-1}-V_t)^2}{2(\sigma_{w}^2 + c_{H}c_{X_t}-q_{H}q_{X_t})}}}{\sqrt{2 \pi (\sigma_{w}^2 + c_{H}c_{X_t}-q_{H}q_{X_t})}}.
\end{multline}

In the $\sft$-phase, i.e., $t=1$, the pilot matrix $\qX_1$ is known. Thus, we substitute $\mse_{X_1} = 0$ into the above expressions.
\end{Proposition}

\begin{proof}
An outline proof is given in the appendix.
\end{proof}

The above result reveals that in the large-system limit, the performance of the quantized MIMO system employing the Bayes-optimal JCD estimator can be fully characterized by the equivalent scalar Gaussian channels (\ref{eq:scalCh_H}) and (\ref{eq:scalCh_X}). For example, the achievable rate under the separate decoding (SD) is the mutual information between $Y_{\tq_{{X_{t}}}}$ and $X_{t}$ for the scalar Gaussian channel (\ref{eq:scalCh_X}). Note that in contrast to joint detection and decoding, the SD involves the joint multiuser detection followed by a bank of independent decoders. Also, the corresponding MSE w.r.t. $\qH$ can be evaluated through the scalar Gaussian channel (\ref{eq:scalCh_H}). Specifically, if the signal is drawn from a quadrature phase shift keying (QPSK) constellation, the corresponding bit error rate (BER) reads
\begin{equation} \label{eq:BER}
    P_{e} = {\cal Q}\left(\sqrt{\tq_X}\right),
\end{equation}
where ${\cal Q}(x) \triangleq \int_{x}^{\infty}\rmD z$ is the Q~function, the MSE w.r.t. payload data $\qX_2$ is given by
\begin{equation} \label{eq:mseQPSK}
    \mse_{X_2} = 1-{\int \rmD z \tanh\left(\tq_{X_2}+\sqrt{\tq_{X_2}}z\right)},
\end{equation}
and the corresponding MSE w.r.t.~$\qH$ is $\mse_{H} = \frac{\sigma_{h}^2}{1+\sigma_{h}^2\tq_{H}}$.

If the channel matrix $\qH$ is perfectly known, the $\sft$-phase is not required so $\beta_2 = \beta$ and $\beta_1 = 0$. Since there is only one phase in $\qX$, we omit the phase indices $(t)$ from all the concerned parameters in this case. Because $\qH$ is perfectly known, we set $\mse_{H} =0$. Plugging this into (\ref{eq:qH}), we immediately obtain $q_{H} = c_{H} = \sigma_{h}^2$, which leads to $q_{H}q_{X} = c_{H} q_{X}$, and ${c_{H}c_{X}-q_{H}q_{X}} = {c_{H} (c_{X}-q_{X})} = {c_{H} \mse_{X}}$. It turns out that the equivalent signal-to-interference-plus-noise ratio (SINR) of the scalar Gaussian channel (\ref{eq:scalCh_X}) is given by
\begin{equation} \label{eq:effNoiseLevelX_PerfCh}
  \tq_X = \alpha c_{H} \left(\sum_{b=1}^{2^{\sfB}-1} \int \rmD v \frac{\left(\Psi'_{b}\left(\sqrt{c_{H} q_{X}} v \right)\right)^2}{\Psi_{b}\left(\sqrt{c_{H} q_{X}} v \right)} \right),
\end{equation}
and the asymptotic MSE w.r.t.~$\qX$ is given by $\mse_{X} = \Ex\{|X- \Ex\{ X | Y_{\tq_{{X}}} \}|^2\}$. If QPSK is used, the MSE in (\ref{eq:mseQPSK}) together with $\tq_X$ in (\ref{eq:effNoiseLevelX_PerfCh}) agree with \cite[(7) \& (8)]{Nakamura-08ISITA}. More precisely, in \cite{Nakamura-08ISITA}, the real-valued system with BPSK signal is considered. In this case, $\sqrt{2}r_{b}$ in our paper should be replaced by $r_{b}$.

\section*{\sc V. Numerical Results}

\begin{figure}
\begin{center}
\resizebox{3.65in}{!}{%
\includegraphics*{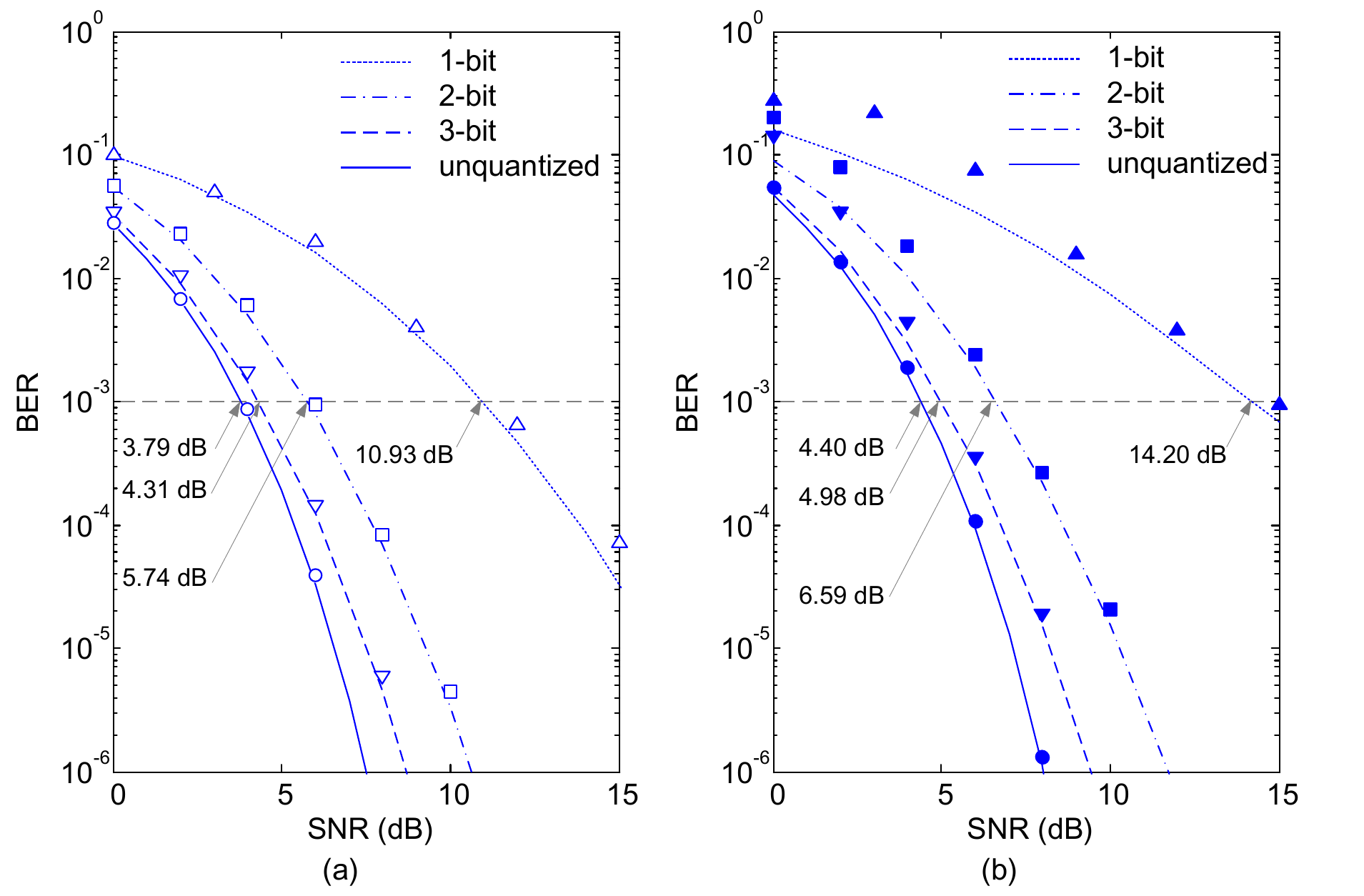} }%
\caption{BER versus SNR${=1/\sigma_{w}^2}$ for QPSK constellations. In the results, the JCD estimation scheme is used under the settings with a) perfect CSIR and b) no CSIR. Curves denote analytical results and markers denote Monte-Carlo simulation results achieved by the GAMP-based JCD algorithm.  }\label{fig:BER_QPSK}
\end{center}
\end{figure}

To verify the accuracy of our analytical results, we compare the analytical BER expression (\ref{eq:BER}) with that obtained by simulations (performed by the GAMP-based JCD algorithm) under the quantized MIMO system with QPSK constellation. The simulation results are obtained by averaging over $10,000$ channel realizations. The parameters of the system are set as follows: $K=50$, $N=200$, $T_1 = 50$, and $T_2 = 450$. The pilot sequences of length $T_1$ are randomly generated. In all the following simulations, we use the typical uniform quantizer with the quantization step-size ${\Delta = \sqrt{0.25}}$. Note that we do not optimize the quantization step-size but select a good one for general scenarios. We leave the related issue to our future work. Figure \ref{fig:BER_QPSK} shows the corresponding BERs results for the cases of 1) perfect CSIR and b) no CSIR. We observe that the analytical BER expression (\ref{eq:BER}) generally predicts well the behavior of the GAMP-based JCD algorithm. For the case with no CSIR, the GAMP-based JCD algorithm cannot work as well as that predicted by the analytical result at low SNRs. This would be because the GAMP-based JCD algorithm is only an approximation to the Bayes-optimal JCD estimator. This gap has motivated the search for other improved estimators in the future. In addition, from Figure \ref{fig:BER_QPSK}, we see that the performance degradation due to $3$-bit quantization is small. For instance, if we target the SNR to that attained by the unquantized system with perfect CSIR at BER$=10^{-3}$, the $3$-bit Bayes-optimal JCD estimator only incurs a loss of $1.19$ dB. Even with $2$-bit quantization, the loss of $2.8$ dB remains acceptable.

\begin{figure}
\begin{center}
\resizebox{3.5in}{!}{%
\includegraphics*{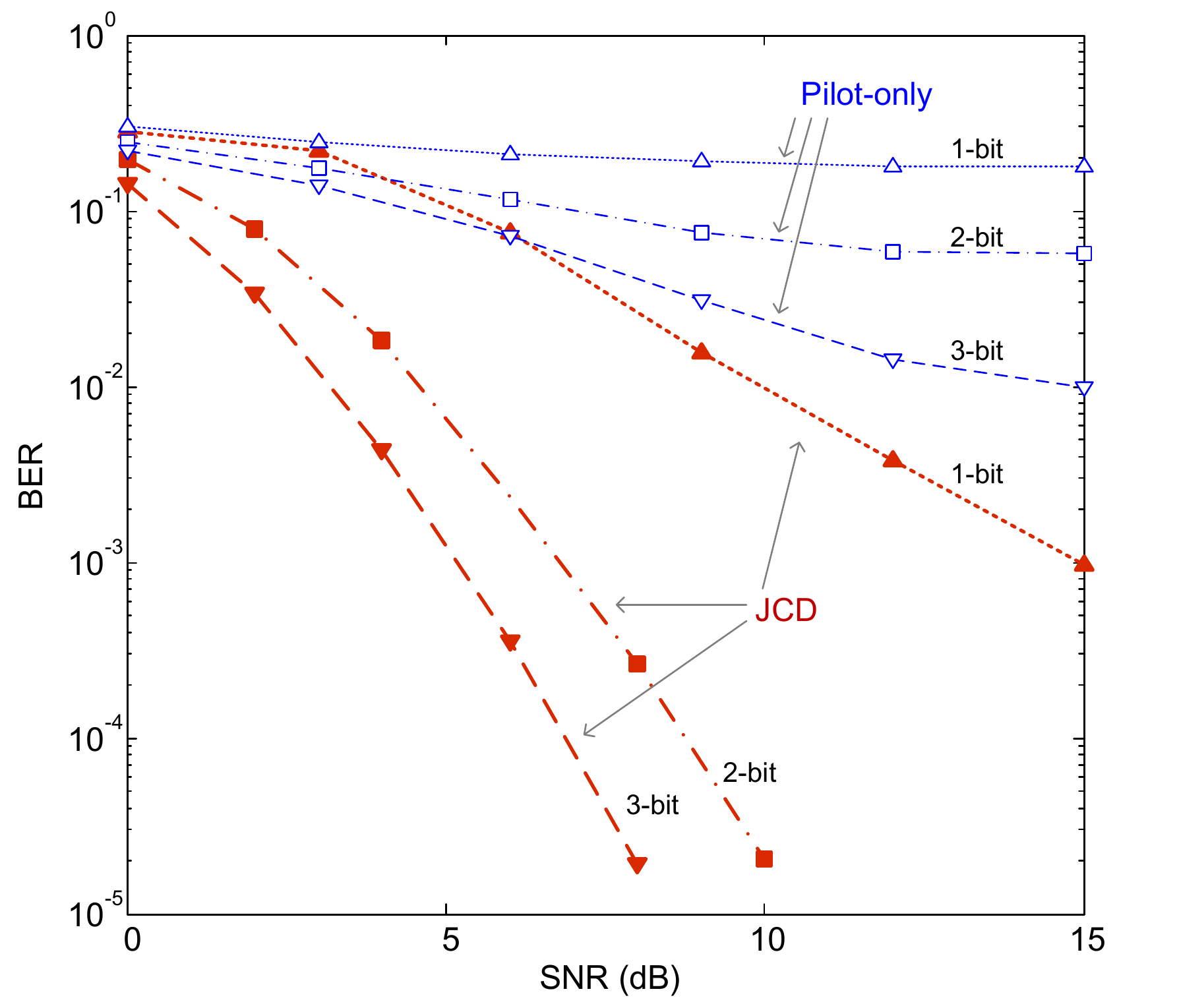} }%
\caption{Average BER versus SNR for QPSK constellations. In the results, the GAMP-based JCD algorithm and pilot-only scheme are used. Plots are based on Monte-Carlo simulation results.  }\label{fig:BER_QPSK_JCDvsPilotOnly}
\end{center}
\end{figure}

Comparing Figures \ref{fig:BER_QPSK}(a) and \ref{fig:BER_QPSK}(b), we see that the loss due to no CSIR is small for the proposed JCD estimator. Therefore, it is of interest to evaluate the improvement due to the JCD estimation. Following the same system parameters as before, Figure \ref{fig:BER_QPSK_JCDvsPilotOnly} compares the BERs under the conventional pilot-only scheme and the proposed JCD estimation scheme. For the pilot-only scheme, we adopt the receiver structure of \cite{Risi-14ArXiv}, which employs the least squares (LS) channel estimate for the quantized MIMO system. However, unlike \cite{Risi-14ArXiv}, we then employ the GAMP algorithm for payload data detection based on the estimated channel. Therefore, the BERs of the pilot-only scheme shown in Figure \ref{fig:BER_QPSK_JCDvsPilotOnly} are expected to be better than that employing suboptimal criteria \cite{Risi-14ArXiv}. Even so, as can be seen from Figure \ref{fig:BER_QPSK_JCDvsPilotOnly}, the JCD estimation still shows a large improvement over the pilot-only scheme.

\begin{figure}
\begin{center}
\resizebox{3.65in}{!}{%
\includegraphics*{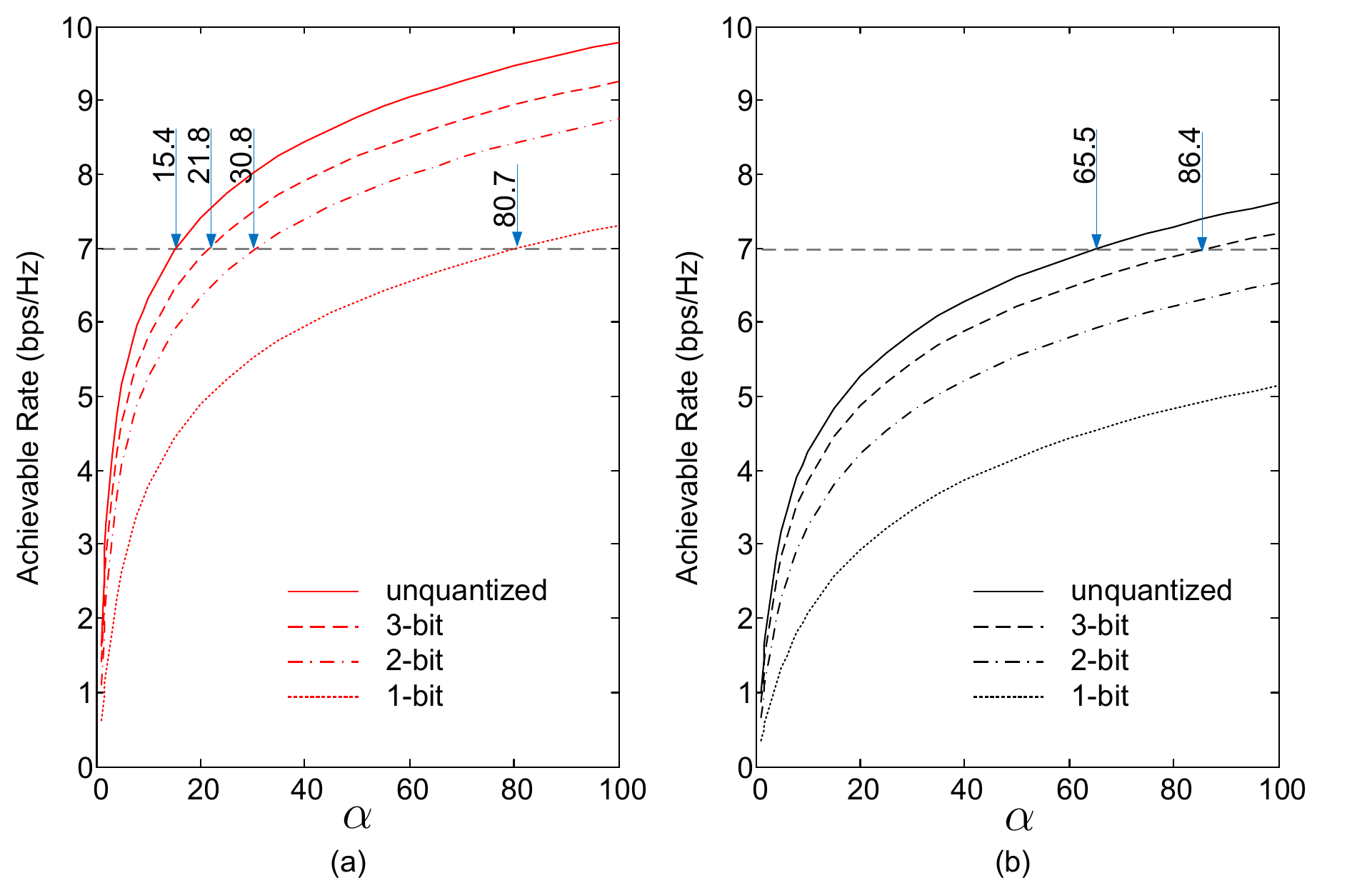} }%
\caption{The achievable rates as functions of $\alpha =N/K$ for a) the Bayes-optimal JCD estimator and b) the pilot-only scheme. $\beta = 10$, $\beta_1 = 1$, $\sigma_{w}^2 = 10^{-1}$, and $X_2^{ij} \sim \calN_{\bbC}(0,1)$.  }\label{fig:Rate_Gaussian}
\end{center}
\end{figure}

Finally, we compare the achievable rates as functions of the antenna ratio $\alpha = N/K$ for the Bayes-optimal JCD estimator and the pilot-only scheme under different quantization precisions in Figure \ref{fig:Rate_Gaussian}. Note that unlike the QPSK signals used in pervious simulations, we consider Gaussian signal, i.e., $X_2 \sim \calN_{\bbC}(0,1)$, in this experiment. We observe that the achievable rates of all the quantization precisions increase as the receive antenna numbers even for the $1$-bit receivers. This implies that the use of high-order modulation schemes is also possible in $1$-bit MIMO systems, which shares the same view as \cite{Mo-14ArXiv}. This property is quite different from the quantized SISO system, where the achievable rate is always upper bounded by $\sfB$ bits in a $\sfB$-bit receiver \cite{Singh-09TCOM}. If we fix the achievable rate to $7$ bps/Hz, the number of receive antennas for the $3$-bit Bayes-optimal JCD estimator is only about $\frac{21.8}{15.4} \approx 1.4$ times of the unquantized Bayes-optimal JCD  estimator (the benchmark receiver), while even with unquantized receivers, the number of receive antennas for the pilot-only scheme requires about $\frac{65.5}{15.4} \approx 4.3$ times of the benchmark receiver. The penalty of increasing $1.4$ times of antenna numbers with very low precision ADCs seems to be quite acceptable.

\section*{VI. Conclusion}
We have developed a framework for studying the best possible estimation performance of the quantized MIMO system. In particular, we used the Bayes-optimal inference for the JCD estimation and realized this estimation by applying the BiG-AMP technique. Additionally, the asymptotic performances (e.g., MSEs) w.r.t.~the channels and the payload data are derived in the lareg-system limit. A set of Monte-Carlo simulations was conducted to illustrate that our analytical results provide an accurate prediction for the performances of the Bayes-optimal JCD estimator. The numerical results have also revealed that the JCD estimation scheme provides tremendous improvement over the conventional pilot-only scheme.

\section*{\sc Appendix A: Proof of Proposition \ref{Pro_MSE}}
As stated in Section IV, the MSEs of interest are saddle points of the average free entropy (\ref{eq:FreeEn}). However, direct calculation is very difficult. Thus, we resort to the replica method by computing the replicate partition function $\Ex_{\wtqY}\left\{\sfP^{\tau}(\qY)\right\}$ in (\ref{eq:LimF}), which with the definition of (\ref{eq:posteriorPr}) can be expressed as
\begin{equation} \label{eq:sf_E1}
    \Ex_{\wtqY}\left\{\sfP^{\tau}(\qY)\right\}
    = \Ex_{\calqH,\calqX}\left\{\int \rmd \wtqY
    \prod_{a=0}^{\tau} {\sfP_{\sf out}\left(\wtqY \left|\qZ^{(a)}\right.\right)} \right\},
\end{equation}
where we define ${\qZ^{(a)} \triangleq \qH^{(a)} \qX^{(a)}/\sqrt{K}}$ with $\qH^{(a)}$ and $\qX^{(a)}$ being the $a$-th replica of $\qH$ and $\qX$ respectively, and the notations $\calqX \triangleq \{ \qX^{(a)}, \forall a \}$ and $\calqH \triangleq \{ \qH^{(a)}, \forall a \}$. Here, $(\qH^{(a)},\qX^{(a)})$ are random matrices taken from the distribution $(\sfP_{\sfH},\sfP_{\sfX})$ for $a=0,1, \dots, \tau$. In addition, $\int \rmd \wtqY$ denotes the integral w.r.t.~a discrete measure because the quantized output $\wtqY$ is a finite set. Next, we will focus on calculating the right-hand side of (\ref{eq:sf_E1}), which can be done by applying the techniques in \cite{Krzakala-13ISIT,Kabashima-14ArXiv,Wen-15ICC} after additional manipulations.

First, in order to average over $(\calqH,\calqX)$, we introduce two ${(\tau+1)\times(\tau+1)}$ matrices $\qQ_{H} = [Q_{H}^{ab}] $ and $\qQ_{X_{t}} =
[Q_{X_{t}}^{ab}] $ whose elements are defined by $Q_{H}^{ab} = \qh_{n}^{(b)} (\qh_{n}^{(a)})^{\dag}/K $ and $Q_{X_{t}}^{ab} = (\qx_{j}^{(a)})^{\dag}
\qx_{j}^{(b)}/K$. Here, $\qh_{n}^{(a)}$ denotes the $n$th row vector of $\qH^{(a)}$, and $\qx_{j}^{(a)}$ denotes the $j$th column vector of $\qX^{(a)}$
corresponding to phase block $t$, and $\calT_t$ for $t=1,2$ represents the set of all symbol indices in phase block $t$. The definitions of $\qQ_{H}$ and $\qQ_{X_{t}}$ are equivalent to
\begin{align*} 
1 &=  \int \prod_{n=1}^{N} \prod_{0\leq a \leq b}^{\tau}\delta\left(
 \qh_{n}^{(b)} (\qh_{n}^{(a)})^{\dag} - K Q_{H}^{ab} \right) \rmd Q_{H}^{ab}, \\
1 &= \int \prod_{t=1}^{2}\prod_{j \in \calT_t}\prod_{0\leq a \leq b}^{\tau}\delta\left(
 (\qx_{j}^{(a)})^{\dag} \qx_{j}^{(b)} - K Q_{X_{t}}^{ab} \right)  \rmd Q_{X_{c,t}}^{ab}, 
\end{align*}
where $\delta(\cdot)$ denotes Dirac's delta. Let $\calqQ_X \triangleq \{ \qQ_{X_t}, \forall t\}$ and $\calqZ  \triangleq \{ \qZ^{(a)}, \forall a\}$. Inserting the
above into (\ref{eq:sf_E1}) yields
\begin{equation}\label{eq:sf_E2}
 \Ex_{\wtqY}\{\sfP^{\tau}(\wtqY)\} = {\int e^{K^2{\cal G}^{(\tau)}}\rmd\mu_H^{(\tau)}(\qQ_H) \rmd\mu_X^{(\tau)}(\calqQ_X)},
\end{equation}
where
\begin{align*} 
{\cal G}^{(\tau)}(\calqQ_Z)&\triangleq \frac{1}{K^2} \log \Ex_{\calqZ }\left\{  \int \rmd\wtqY\prod_{a=0}^{\tau} {\sfP_{\sf out}\left(\wtqY \left| \qZ^{(a)} \right.\right)}  \right\}, \\
\mu_H^{(\tau)}(\qQ_H) &\triangleq\Ex_{\calqH}\left\{\prod_{n,a,b} \delta\left(
 \qh_{n}^{(b)} ( \qh_{n}^{(a)} )^{\dag} - K Q_{H}^{ab} \right)\right\},\\
\mu_X^{(\tau)}(\calqQ_X) &\triangleq\Ex_{\calqX}\left\{\prod_{t,j,a,b}\delta\left(
 (\qx_{j}^{(a)})^{\dag} \qx_{j}^{(b)} - K Q_{X_{t}}^{ab} \right)\right\}.
\end{align*}
Then using the Fourier representation of the $\delta$ function and computing the integrals by the saddle point method, we attain
\begin{equation}
\frac{1}{K^2}\Ex_{\wtqY}\{\sfP^{\tau}(\wtqY)\} = \Extr_{\qQ_H,\calqQ_X,\tilde{\qQ}_H,\tilde{\calqQ}_X} \{ \Phi^{(\tau)} \}
\end{equation}
with
\begin{align*} 
&\hspace{-0.2cm} \Phi^{(\tau)} \triangleq \frac{1}{K^2} \log \Ex_{\qZ}\left\{  \prod_{n,t,j \in \calT_t}  \int \rmd y_{n,j} \prod_{a} {\sfP_{\sf out}\left(y_{n,j} \left| z_{n,j}^{(a)} \right.\right)}  \right\}  \notag \\
&\hspace{-0.2cm} +\frac{1}{K^2} \log\Ex_{\calqH}\left\{\prod_{n} e^{\tr\left(\tilde\qQ_{H}\qH_{n}^{\dag}\qH_{n}\right)}\right\} -\alpha {\tr\left(\tilde\qQ_{H}\qQ_{H}\right)} \notag \\
&\hspace{-0.2cm} +\frac{1}{K^2}  \log\Ex_{\calqX}\left\{\prod_{t}e^{\tr\left(\tilde\qQ_{X_{t}}\qX_{t}^{\dag} \qX_{t}\right)}\right\} -{\sum_{t}\beta_t
 \tr\left(\tilde\qQ_{X_{t}}\qQ_{X_{t}}\right)},
\end{align*}
where $\Extr_{x}\{ f(x) \}$ represents the extreme value of $f(x)$ w.r.t.~$x$, $\tilde{\qQ}_{H} = [\tilde{\qQ}_{H}^{ab}] \in \bbC^{(\tau+1)\times(\tau+1)}$ and $\tilde{\calqQ}_X \triangleq\{ \tilde{\qQ}_{X_{t}} = [\tilde{\qQ}_{X_{t}}^{ab}] \in \bbC^{(\tau+1)\times(\tau+1)}, \forall c,t \} $. According to (\ref{eq:LimF}), the average free entropy turns out to be $\Phi = \lim_{\tau\rightarrow 0}\frac{\partial}{\partial \tau}\Extr_{\qQ_H,\calqQ_X,\tilde{\qQ}_H,\tilde{\calqQ}_X} \{ \Phi^{(\tau)} \}$.

The saddle points of $\Phi^{(\tau)}$ can be obtained by seeking the point of zero gradient w.r.t.~$\{\qQ_{H},\qQ_{X_{t}},\tilde{\qQ}_{H},\tilde{\qQ}_{X_{t}}\}$. However, in doing so, it is prohibitive to get explicit expressions about the saddle points. Therefore, we assume that the saddle points follow the RS form \cite{Nishimori-01BOOK} as $\qQ_{H} = (c_{H}-q_{H})\qI +q_{H}\qone$, $\tilde\qQ_{H} =  (\tc_{H}-\tq_{H})\qI + \tq_{H}\qone$, $\qQ_{X_{t}} =(c_{X_{t}}-q_{X_{t}})\qI + q_{X_{t}}\qone$, $\tilde\qQ_{X_{t}} = (\tc_{X_{t}}-\tq_{X_{t}})\qI + \tq_{X_{t}}\qone$, where $\qI$ denotes the identity matrix and $\qone$ denotes the all-one matrix. In addition, the application of the central limit theorem suggests that the $\qz_{n,j} \triangleq [ z_{n,j}^{(0)} \, z_{n,j}^{(1)} \cdots z_{n,j}^{(\tau)} ]^T$ are Gaussian random vectors with ${(\tau+1)\times(\tau+1)}$ covariance matrix $\qQ_{Z_t}$. If $j \in \calT_t$, then the $(a,b)$th entry of $\qQ_{Z_t}$ is given by
\begin{equation} \label{eq:defQ}
     (z_{n,j}^{(a)})^* z_{n,j}^{(b)}
     = Q_{H}^{ab} Q_{X_{t}}^{ab} \triangleq Q_{Z_t}^{ab} .
\end{equation}
Therefore, we set $\qQ_{Z_t} = (c_{H}c_{X_{t}}-q_{H}q_{X_{t}})\qI + q_{H}q_{X_{t}}\qone $, which is equivalent to introduce to the Gaussian random variable $\qz_{n,j}$ for $j \in \calT_t$ as
\begin{equation}
    z_{n,j}^{(a)} = \sqrt{ c_{H}c_{X_{t}}-q_{H}q_{X_{t}} } \,u^{(a)}
    + \sqrt{ q_{H}q_{X_{t}} }
\end{equation}
for $a=0,1, \dots \tau$, where $u^{(a)}$ and $v$ are independent standard complex Gaussian random variables.

Substituting these RS expression into $\Phi^{(\tau)}$ leads to the RS expression of $\Phi^{(\tau)}$. With the RS, we only have to determine the parameters $\{c_{H},q_{H},c_{{X_{t}}},q_{{X_{t}}},\tc_{H},\tq_{H},\tc_{{X_{t}}},\tq_{{X_{t}}} \}$, which can be obtained by equating the corresponding partial derivatives of $\Phi^{(\tau)}$ to zero. In doing so, as $\tau \rightarrow 0$, we get that $\tc_{H} =0$, $\tc_{{X_{t}}} = 0$, $c_{H} =\Ex\{|{H}|^2\}$, $c_{{X_{t}}} =\Ex\{|{X_{t}}|^2\}$, and the other parameters $\{q_{H},q_{X_{t}},\tq_{H},\tq_{X_{t}} \}$ are given by (\ref{eq:fxiedPoints}) in Proposition \ref{Pro_MSE}. Let $\mse_{H} = c_{H}-q_{H}$ and $\mse_{{X_{t}}} = c_{{X_{t}}}-q_{{X_{t}}}$. After taking these into account, we obtain the RS expression of the average free entropy as
\begin{align}
 \Phi &= \alpha \sum_{t} \beta_t \Bigg( \sum_{b} \int \rmD v\,
   \Psi_{b}\left( V_{t} \right)
   \log \Psi_{b}\left( V_{t} \right)
    \Bigg) \notag \\
  & \quad - \alpha  I(H;Z_{H}|\tq_{H}) - \sum_{t=1}^{2}  \beta_t  I(X_{t};Z_{X_{t}}|\tq_{X_{t}}) \notag \\
  & \quad + \alpha (c_{H} - q_{H})\tq_{H} + \sum_{t=1}^{2} \beta_t (c_{X_{t}}-q_{X_{t}}) \tq_{X_{t}}. \label{eq:FreeEntropyFinal}
\end{align}
where we have defined $V_{t} \triangleq \sqrt{q_{H}q_{X_{t}}} v$, $\Psi_{b}(V_{t})$ is given by (\ref{eq:Pout}), and the notation $I(X,Z|q)$ is used to denote the mutual information between $X$ and $Z$ with a Gaussian scalar channel ${Z = \sqrt{q} X + W}$ and $W \sim \calN_{\bbC}(0, 1)$. Note that the parameters $\{q_{H},q_{X_{t}},\tq_{H},\tq_{X_{t}} \}$ given by (\ref{eq:fxiedPoints}) are  saddle points of (\ref{eq:FreeEntropyFinal}).

{\renewcommand{\baselinestretch}{1.1} 
\bibliographystyle{IEEEtran}

\end{document}